\documentclass[twocolumn,10pt]{asme2e}

\usepackage{amsmath,graphicx,amsfonts,amssymb,epsfig,subfigure,mathrsfs,mathtools}
\usepackage{color}
\usepackage{multirow}
\usepackage{rotating}
\usepackage{graphicx}%
\usepackage{algorithm}
\usepackage{algpseudocode}
\usepackage{tikz}

\newcommand{\tr}{\text{ tr}}

\newtheorem{assumption}{Assumption}
\newtheorem{theorem}{Theorem}

\newtheorem{definition}{Definition}
\newtheorem{lemma}{Lemma}
\newtheorem{proposition}{Proposition}
\newtheorem{remark}{Remark}

\confshortname{DSCC 2014}
\conffullname{the ASME 2014 Dynamic Systems and Control Conference}

\confdate{October 22-24}
\confyear{2014}
\confcity{San Antonio}
\confcountry{TX, USA}

\papernum{DSCC2014-5877}

\title{Optimal Switching Synthesis for Jump Linear Systems with Gaussian initial state uncertainty}

\author{Kooktae Lee\thanks{Address all correspondence to this author.}
    \affiliation{
    Department of Aerospace Engineering\\
    Texas A\&M University\\
    College Station, Texas 77843-3141, USA\\
    Email: animodor@tamu.edu
    }	
}

\author{Raktim Bhattacharya
	\affiliation{
    Department of Aerospace Engineering\\
    Texas A\&M University\\
    College Station, Texas 77843-3141, USA\\
    Email: raktim@tamu.edu
    }
}

\begin{document}

\maketitle    

\begin{abstract}
{\it This paper provides a method to design an optimal switching sequence for jump linear systems with given Gaussian initial state uncertainty. In the practical perspective, the initial state contains some uncertainties that come from measurement errors or sensor inaccuracies and we assume that the type of this uncertainty has the form of Gaussian distribution. In order to cope with Gaussian initial state uncertainty and to measure the system performance, Wasserstein metric that defines the distance between probability density functions is used. Combining with the receding horizon framework, an optimal switching sequence for jump linear systems can be obtained by minimizing the objective function that is expressed in terms of Wasserstein distance. The proposed optimal switching synthesis also guarantees the mean square stability for jump linear systems. The validations of the proposed methods are verified by examples.}
\end{abstract}

\begin{nomenclature}
\entry{$\|\cdot\|$}{Without subscription denotes $\ell_2$-norm}
\entry{$\mathbb{R}^{+}$}{The set of positive real}
\entry{$\mathbb{Z}^+$}{The set of non-negative integer}
\entry{$\mathcal{I}$}{The set of switching modes}
\entry{$\text{tr}\left(\cdot\right)$}{Trace operator for a square matrix}
\entry{m.s.}{Convergence in the mean square sense}
\entry{$X \sim \varsigma\left(x\right)$}{random variable $X$ with probability density function (PDF) $\varsigma\left(x\right)$}
\entry{$\mathcal{N}\left(\mu,\Sigma\right)$}{Gaussian PDF with mean $\mu$ and covariance $\Sigma$}
\end{nomenclature}

\section*{INTRODUCTION}
A jump linear system is defined as a dynamical system consisting of a finite number of subsystems and a switching rule that governs a switching between the family of linear subsystems. Over decades, 
a variety of researches for jump linear systems have been investigated because of its practical implementation. For example, a jump linear system can be used for power systems, manufacturing systems, aerospace systems, networked control systems, etc(\cite{boukas2006stochastic},\cite{cassandras2001optimal},\cite{do2005discrete}). 

In general, problems for jump linear systems branch out into two different fields. The first one is the stability analysis under given switching laws. Since a certain switching law between individually stable subsystem can make the jump linear system unstable\cite{liberzon2003switching}, it is very important to identify conditions under which system can be stable. Interestingly, the jump linear system also can be stable by switching between unstable subsystems. Fang \textit{et al.}\cite{fang2004stabilization} showed sufficient conditions for stability of jump linear systems under arbitrary switching using linear matrix inequalities(LMIs).
Lin \textit{et al.}\cite{lin2009stability} showed necessary and sufficient conditions for asymptotic stability of jump linear systems using finite n-tuple switching sequences, satisfying a certain condition. In addition, broad analysis regarding stability for jump linear systems has been accomplished within few decades(\cite{feng1992stochastic,liberzon1999basic,hespanha1999stability,sun2005analysis,lee2014acc,lee2014robustness}).

On the other hand, switching synthesis problem, which is another branch of jump linear systems, is relatively new and few investigations have been achieved. 
Since the main objective is to design switching sequences that establish both the stability and the performance, this case is much harder than stability analysis problem. For instance, Das and Mukherjee\cite{das2008optimally} solved the problem for an optimal switching of jump linear systems using Pontryagin's minimum principle. In this method, two-point boundary value problem was solved via relaxation method, where ordinary differential equations are approximated by finite difference equations on mesh points. Therefore, the optimality and computational cost depend on mesh size. In addition, the time to find optimal solution varies according to guess solution. 
Egerstedt \textit{et al.}\cite{egerstedt2003optimal} addressed a method to find derivative of the cost function with respect to switching time. However, in this paper, switching sequences are already given and the main focus is to find switching time. Although several other researches regarding optimal control problem together with optimal switching were studied for switched nonlinear systems(\cite{hedlund1999optimal,xu2004optimal,bengea2005optimal}), they may not fit to pure optimal switching problem for jump linear systems.

Here we address optimal switching problem for jump linear systems with given multi-controllers. Multi-controller switching scheme is widely used, such as plant stabilization\cite{minto1991new}, system performance\cite{lin2009stability}, adaptive control\cite{narendra1994improving}, and resource-constrained scheduling\cite{boctor1990some}. Under the assumption that more than two controllers are given to user, our objective is to find the optimal switching sequence which attains the best performance of the system by controller switching. We can also extend our method to multi-model switching problem by generalizing the multi-controller switching problem. Consequently, we aim to synthesize switching protocols that result in the optimality for the performance of jump linear systems. 
Moreover, we address the optimal switching problem with initial state uncertainties. In the practical perspective, initial state may contain uncertainties that  usually come from measurement errors or sensor inaccuracies. Then, the system state is expressed as random variables represented by PDFs. We assume that the initial state PDF has a form of Gaussian distribution that is very common for real implementation. In order to measure the performance of the jump linear system with a given Gaussian PDF, we need to adopt a proper metric. In this paper, Wasserstein metric that assesses the distance between PDFs is used as a tool for both the stability and the performance measure. Hence, we introduce the optimal switching synthesis to achieve the optimality of the system performance with given Gaussian initial PDF by minimizing the objective function that is expressed in terms of Wasserstein distance. We also prove that the convergence of Wasserstein distance implies the mean square stability for the jump linear systems.


Rest of this paper is organized as follows. Section II introduces the problem we want to solve. Brief explanations of Wasserstein distance are described in Section III. Section IV provides a way to solve optimal switching problems using receding horizon framework when Gaussian initial state uncertainty exists. Then, Section V demonstrates the validation of proposed methods by examples and Section VI concludes this paper.


\section*{PROBLEM STATEMENT}
Consider a discrete-time linear system with multi-controller, given by 
\begin{eqnarray}
&x(k+1) = Ax(k)+Bu_{\sigma_k}(x),\quad k\in\mathbb{Z}^{+}, \sigma_k\in\mathcal{I}\label{1}\\
&u_{\sigma_k}(x) = K_{\sigma_k}x\label{2}
\end{eqnarray}
where the state vectors $x\in \mathbb{R}^{n}$, control inputs $u_{\sigma}\in\mathbb{R}^{m}$, the system matrices $A\in\mathbb{R}^{n\times n}$, $B\in\mathbb{R}^{n\times m}$, and the set of modes $\mathcal{I}=\{1,2,\cdots,m\}$. Note that the system matrix $A$ is time-invariant and user can select one controller $K_{\sigma_k}$ out of multiple choices. Without loss of generality, we can convert system \eqref{1}-\eqref{2} to the following jump linear systems by letting $A_{\sigma_k} := A+BK_{\sigma_k}$.
\begin{eqnarray}
x(k+1) = A_{\sigma_k}x(k),\quad k\in\mathbb{Z}^{+}, \sigma_k\in\mathcal{I}\label{3}
\end{eqnarray}
where the system matrices $A_{\sigma_{k}}\in\mathbb{R}^{n\times n}$. 

The system in \eqref{3} represents not only the controller switching as depicted in \eqref{1}-\eqref{2}, but also the system mode switching. Hence, we consider the jump linear system model \eqref{3} and we assume that individual subsystem dynamics $A_{\sigma_k}$ are Schur stable. Our objective is to find the switching sequence, $\sigma = \{\sigma_1, \sigma_2, \cdots \}$, which guarantees the optimal performance of the switched system. For example, with multi-controller, we want to design a switching law which makes the system states reach the origin as fast as possible. Therefore, our aim is not to design controllers, but rather to synthesize the optimal switching sequence.

For simplicity, we assume that there are two different controllers, which are good and poor in terms of system performance. The closed-loop dynamics are given by $A_1$ and $A_2$, respectively. In general, the reason to design multi-controller with respect to single system is to attain not only the system performance but also system stability, robustness, resource-optimal scheduling, etc.

The convergence rate of system state is determined by spectral radius $\rho(A_\sigma) :=\max_{j}|\lambda_{\sigma}^j|$, where $\lambda_{\sigma} = \{\lambda_{\sigma}^1, \lambda_{\sigma}^2, \cdots, \lambda_{\sigma}^n\}$ is the set of eigenvalues for $A_{\sigma}$ mode. According to characteristics of subsystem $A_{\sigma}$, there may exist a surge or an elevation in the state trajectory. In Fig. \ref{fig:schematic}, we show one possibility where the switching is necessary for better performance of the system. Solid line represents the state trajectory of $A_1$ while dashed line shows that of $A_2$. In contrast to $A_2$, which has slow convergence rate with no surge, $A_1$ reaches the origin faster with a surge. Therefore, for better performance, it is clear that $A_2$ mode has to be used from the beginning, and then system has to switch to $A_1$ mode at time $t_k$ as described in arrows in Fig. \ref{fig:schematic}.

\begin{figure}[h!] 
\begin{center}
\includegraphics[width=0.45\textwidth]{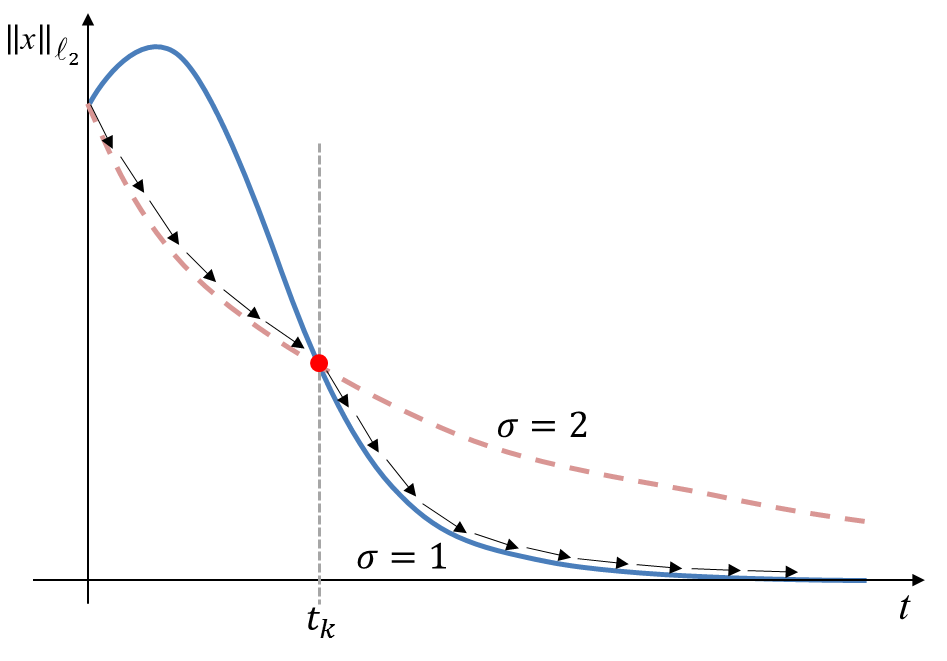}
\vspace{-0.2in}
\caption{Schematic of Optimal Switching for Jump System}
\label{fig:schematic}
\end{center}
\end{figure}

In this paper, motivated by the above example we address the following two questions.
\begin{enumerate}
\item Is there a switching sequence for a jump linear system to get better performance compared to single mode?
\item If the above holds true, can we find the optimal switching sequence?
\end{enumerate}

In general, it is difficult to answer the first question directly. Instead, we want to show the case where the switching synthesis is not required because single mode attains the best performance. 
When $\rho(A_1) < \rho(A_2)$, $A_1$ mode has faster convergence to the origin than $A_2$ mode. In addition, if $\|A_1 x(k)\| < \|A_2 x(k) \|$ for all $k$, then $\|x(k)\|$ using $A_1$ mode is always less than $\|x(k)\|$ using $A_2$ mode. As a result, $A_1$ mode attains the best performance and jump is not necessary. 
%

For the second question, which is the main contribution of this paper, we introduce the optimal switching sequence using receding horizon framework and it is explained in section IV. Since, in most cases, initial condition of system state contains uncertainties, which come from measurement errors or sensor inaccuracies, we will use probability for initial state uncertainty of the system. Moreover, we assume that the type of initial state uncertainties is given by Gaussian distribution. The deterministic single initial state is a special case for Gaussian distribution with zero covariance. Therefore, in this paper we conceptually cover much broader one. Due to this Gaussian PDF, system states become a random number, and hence we cannot use $\ell_2$-norm for the performance measure. As a consequence, we need to adopt a proper metric to quantify the distance between PDFs to measure the performance. For this reason, instead of using $\ell_2$-norm $\|\cdot\|_{\ell_2}$, Wasserstein distance is used as a tool for measuring the performance of jump linear systems. Brief explanations of Wasserstein distance are introduced in the next section.


\section*{WASSERSTEIN DISTANCE}
\begin{definition} ({Wasserstein distance})
Consider the metric space $\ell_{2}\left(\mathbb{R}^{n}\right)$ and let the vectors $x_{1}, x_{2} \in \mathbb{R}^{n}$. Let $\mathcal{P}_{2}(\varsigma_{1},\varsigma_{2})$ denote the collection of all probability measures $\varsigma$ supported on the product space $\mathbb{R}^{2n}$, having finite second moment, with first marginal $\varsigma_{1}$ and second marginal $\varsigma_{2}$. Then the $L_{2}$ Wasserstein distance of order 2, denoted as $_{2}W_{2}$, between two probability measures $\varsigma_{1},\varsigma_{2}$, is defined as
\begin{align}\label{Wassdefn}
&_{2}W_{2}(\varsigma_{1},\varsigma_{2}) \triangleq \\ \nonumber & \left(\displaystyle\inf_{\varsigma\in\mathcal{P}_{2}(\varsigma_{1},\varsigma_{2})}\displaystyle\int_{\mathbb{R}^{2n}} \parallel x_{1}-x_{2}\parallel_{\ell_{2}\left(\mathbb{R}^{n}\right)}^{2} \: d\varsigma(x_{1},x_{2}) \right)
^{\frac{{1}}{2}}.
\label{W-dist}
\end{align}
\end{definition}
\begin{remark}
Intuitively, Wasserstein distance equals the \emph{least amount of work} needed to morph one distributional shape to the other, and can be interpreted as the cost for Monge-Kantorovich optimal transportation plan \cite{villani2003topics}. For notational ease, we henceforth denote $_{2}W_{2}$ as $W$. Further, one can prove (p. 208, \cite{villani2003topics}) that $W$ defines a metric on the manifold of PDFs.
\label{WassRemarkFirst}
\end{remark}
Next, we present new results for stability in terms of $W$.
 
\begin{proposition}\label{m.s.stable}
If we fix Dirac distribution as the reference measure, then distributional convergence in Wasserstein metric is \emph{necessary and sufficient} for convergence in m.s. sense.\label{WConvMeanSqConvProposition}
\end{proposition}
\begin{proof}
Consider a sequence of $n$-dimensional joint PDFs $\{\varsigma_{j}\left(x\right)\}_{j=1}^{\infty}$, that converges to $\delta\left(x\right)$ in distribution, i.e., $\displaystyle\lim_{j\rightarrow\infty} W\left(\varsigma_{j}(x), \delta(x)\right) = 0 = \displaystyle\lim_{j\rightarrow\infty} W^{2}\left(\varsigma_{j}(x), \delta(x)\right)$. From (\ref{Wassdefn}), we have
\begin{eqnarray}
&W^{2}\left(\varsigma_{j}(x), \delta(x)\right) = \displaystyle\inf_{\varsigma\in\mathcal{P}_{2}(\varsigma_{j}(x),\delta(x))} \mathbb{E}\left[\parallel X_{j} - 0 \parallel_{\ell_{2}\left(\mathbb{R}^{n}\right)}^{2}\right]\label{ms_stab} \\ \nonumber
& = \mathbb{E}\left[\parallel X_{j} \parallel_{\ell_{2}\left(\mathbb{R}^{n}\right)}^{2}\right]
\end{eqnarray}
where the random variable $X_{j} \sim \varsigma_{j}\left(x\right)$, and the last equality follows from the fact that $\mathcal{P}_{2}(\varsigma_{j}(x),\delta(x)) = \{\varsigma_{j}(x)\}$ $\forall \: j$, thus obviating the infimum. From \eqref{ms_stab}, $\displaystyle\lim_{j\rightarrow\infty} W\left(\varsigma_{j}(x), \delta(x)\right) = 0 \Rightarrow \displaystyle\lim_{j\rightarrow\infty} \mathbb{E}\left[\parallel X_{j} \parallel_{\ell_{2}}^{2}\right] = 0$, establishing distributional convergence to $\delta(x) \Rightarrow$ m.s. convergence. Conversely, m.s. convergence $\Rightarrow$ distributional convergence, is well-known \cite{grimmett2001probability} and unlike the other direction, holds for arbitrary reference measure.
\end{proof}


\begin{proposition}($W^{2}$ between Gaussian and Dirac PDF (see e.g., p. 160-161,\cite{hassani2013mathematical}))
The Wasserstein distance between Gaussian and Dirac PDF supported on $\mathbb{R}^{n}$, with respective joint PDFs $\varsigma = \mathcal{N}\left(\mu,\Sigma\right)$ and $\delta\left(x\right) = \displaystyle\lim_{\mu,\Sigma \rightarrow 0} \mathcal{N}\left(\mu,\Sigma\right)$, is given by,
\begin{eqnarray}\label{staticW}
W^{2}\left(\mathcal{N}\left(\mu,\Sigma\right), \delta\left(x\right)\right) = \parallel \mu \parallel_{\ell_{2}\left(\mathbb{R}^{n}\right)}^{2} + \: \text{tr}\left(\Sigma\right).
\label{GaussianDiracW}
\end{eqnarray}
\label{GaussToDiracCorollary}
\end{proposition}


\section*{SWITCHING SYNTHESIS USING RECEDING HORIZON FRAMEWORK WITH WASSERSTEIN METRIC}
\subsection*{Optimal Switching Problem}

$W^{2}$ defined in \eqref{staticW} represents square Wasserstein distance at fixed time. However, because the state PDF changes over time along dynamics, $W^2$ also changes as time goes. The following proposition expresses time-varying square $W$ distance between $\mathcal{N}(\mu ,\Sigma)$ and $\delta(x)$ at time k.
\begin{proposition}\label{W2(k)}
Let $W^2(k)$ denote square Wasserstein distance between $\mathcal{N}(\mu,\Sigma)$ and $\delta(x)$ at time k. Then $W^2$ distance at time k is given by
\begin{align}
W^{2}(k) = vec(I_n)\prod_{p=1}^{k}\left(A_{\sigma_p}\otimes A_{\sigma_p}\right)vec\left(\mu_0\mu_0^{\top}+\Sigma_0\right)
\end{align}
where $\mu_0$ and $\sigma_0$ are mean and covariance of initial Gaussian PDF.
\end{proposition}

\begin{proof}
From \eqref{staticW}, $W^{2}$ at time $k+1$ is defined as
\begin{eqnarray}
W^{2}(k+1) =& \parallel \mu(k+1) \parallel ^{2} + \tr\left(\Sigma(k+1)\right)\\
=& \tr\left(\mu(k+1)\mu(k+1)^{\top}+\Sigma(k+1)\right)\label{9}
\end{eqnarray}
Note that $\mathcal{N}(\mu(k) ,\Sigma(k))$ remains Gaussian PDF for all time $k$, even in the mode switching between sublinear dynamics. The following property are used for updating mean and covaraince of Gaussian PDF.
\begin{eqnarray}
&\mu(k+1) = A_{\sigma_{k}} \mu(k)\label{mu_update}\\
&\Sigma(k+1) = A_{\sigma_{k}} \Sigma(k) A_{\sigma_{k}}^{\top}\label{sigma_update}
\end{eqnarray}
Substituting \eqref{mu_update} and \eqref{sigma_update} into \eqref{9}, we get
\begin{eqnarray}\label{21}
W^{2}(k+1) =\tr\left(A_{\sigma_{k}}^{\top}A_{\sigma_{k}}\left(\mu(k)\mu(k)^{\top}+\Sigma(k)\right)\right)
\end{eqnarray}
Using $\tr(X^{\top}Y) = vec(X)^{\top}vec(Y)$, \eqref{21} can be expressed as
\begin{align}\label{22}
W^{2}(k+1) =& 
vec(A_{\sigma_{k}}^{\top}I_nA_{\sigma_{k}})^{\top}
vec\left(\mu(k)\mu(k)^{\top}+\Sigma(k)\right)
\end{align}
Further, by applying $vec(ABC) = (C^{\top}\otimes A)vec(B)$ to the first term of right hand side in \eqref{22}, we get

\begin{align}\label{23}
  &\begin{aligned}
    W^{2}(k+1) &= vec(I_n)^{\top}\left(A_{\sigma_k}\otimes A_{\sigma_k}\right)\\
      \qquad\qquad&\qquad\qquad vec\left(\mu(k)\mu(k)^{\top}+\Sigma(k)\right)
  \end{aligned}
\end{align}

Similarly, $W^{2}$ at time k is also obtained as
\begin{eqnarray}\label{24}
W^{2}(k) = vec(I_n)^{\top}vec\left(\mu(k)\mu(k)^{\top}+\Sigma(k)\right)
\end{eqnarray}

From \eqref{23} and \eqref{24}, and by induction, we conclude that $W^{2}(k)$ can be expressed in terms of initial mean and covariance as follows.
\begin{align}\label{25}
W^{2}(k) = vec(I_n)^{\top}\prod_{p=1}^{k} \left(A_{\sigma_p}\otimes A_{\sigma_p}\right) vec\left(\mu_0\mu_0^{\top}+\Sigma_0)\right)
\end{align}
\end{proof}

We aim to find the switching sequence which guarantees the optimality of the system performance. One way of doing that is to minimize the area of Wasserstein distance, and hence minimize the time for the state PDF $\mathcal{N}\left(\mu (k),\Sigma (k)\right)$ to reach the reference PDF $\delta(x)$. In this case, we can formulate the cost function as
\begin{eqnarray}\label{17}
J(\sigma) = \int_{0}^{\infty}W^2 dt = \sum_{k=0}^{\infty}W^2(k)dk
\end{eqnarray}
where $dk$ is a sampling time for discrete-time system. We use discrete-time $W^2$, and hence equality between second and last equations in \eqref{17} holds. From the cost function in \eqref{17}, the optimal switching problem is defined as follows.

\noindent\textbf{Optimal Switching Problem}
\begin{align}\label{Jopt}
J(\sigma^{*})=&\min_{\sigma}J(\sigma)
\end{align}

The solution of the above optimal switching problem can be obtained by finding optimal switching sequence $\sigma^{*}=\{\sigma_1^*, \sigma_2^*, \cdots\}$ out of all switching possibilities. For example, if the terminal time is finite and is set to be $n$ instead of $\infty$ in \eqref{Jopt}, we have to check total $m^n$ switching sequences for optimal solution, where $m$ is total number of modes. Therefore, this problem is same as a conventional tree-search problem\cite{dakin1965tree}. Since the growth of tree size is exponential in time, this problem is extremely difficult to solve and it requires large computational time. More details with respect to issues on complexity are discussed in the last subsection. Therefore, we want to simplify the original problem by the next assumption.
\begin{assumption}\label{assump1}
For the jump linear system in \eqref{2}, switching sequence $\sigma$ is constant over given horizon $T$.
\end{assumption}

Using assumption \ref{assump1}, we can apply the receding horizon framework and the cost function over horizon length $T$ can be defined as
\begin{align}
J &= \sum_{k=t_j}^{T+t_j}W^{2}(k)\label{discrete_W}dk\\
&= \sum_{k=t_j}^{T+t_j}vec(I_n)\left(\prod_{p=1}^{k}\left(A_{\sigma_p}\otimes A_{\sigma_p}\right)\right)vec(\mu_0\mu_0^{\top}+\Sigma_0)\label{J_varying_sigma}dk\\
&= \sum_{k=t_j}^{T+t_j}vec(I_n)\left(\left(A_{\sigma}\otimes A_{\sigma}\right)^k\right)vec(\mu_0\mu_0^{\top}+\Sigma_0)\label{J_fixed_sigma}dk\\
&=\sum_{k=t_j}^{T+t_j}W_{\sigma}^{2}(k)dk
\end{align}

Switching sequence, denoted as $\sigma$, is fixed and we get \eqref{J_fixed_sigma} from \eqref{J_varying_sigma} for $\sigma_p=\sigma=$ constant under the assumption \ref{assump1}.

Then, the optimal cost-to-go function is defined as:

\noindent\textbf{Optimal Switching with Receding Horizon}
\begin{align}
J^{*} &= \min_{\sigma}\left(\sum_{k=t_j}^{T+t_j}W_{\sigma}^{2}(k)dk\right)\label{Jopt2}
\end{align}
\begin{align}
s.t.\quad 
W^{2}_{\sigma}(t_{j+1})&-W^{2}_{\sigma}(t_{j-1}) \leq -\epsilon_{\sigma}(t_j)\label{stab_cond}
\end{align}

where $\epsilon_{\sigma}(\cdot)$ is a positive definite function and the constraint \eqref{stab_cond} is enforced for stability. 

It is well known\cite{liberzon2003switching} that switching between individually stable modes can make a stable system unstable. Therefore, the constraint \eqref{stab_cond} should be enforced to ensure stability. Fig. \ref{fig:jump} shows schematic of optimal switching sequence using receding horizon framework. 
At time $t_j$, the solution of \eqref{Jopt2}-\eqref{stab_cond} provides the optimal switching sequence for horizon $T$ and there is no switching during time $k\in[t_j,t_{j+1})$. When time $k$ reaches $t_{j+1}$, we again compute optimal switching for next horizon $T$.

Note that although \eqref{stab_cond} implies piecewise monotone decreasing in $W^2$, it is not so restrictive condition because \eqref{stab_cond} is only applied to the time $t_j$ at which jump occurs. In other words, $W^{2}(\cdot)$ can increase in between times $k\in[t_j,t_{j+1})$ as depicted in Fig.\ref{fig:jump}.
\begin{figure}[h!] 
\begin{center}
\includegraphics[width=0.5\textwidth]{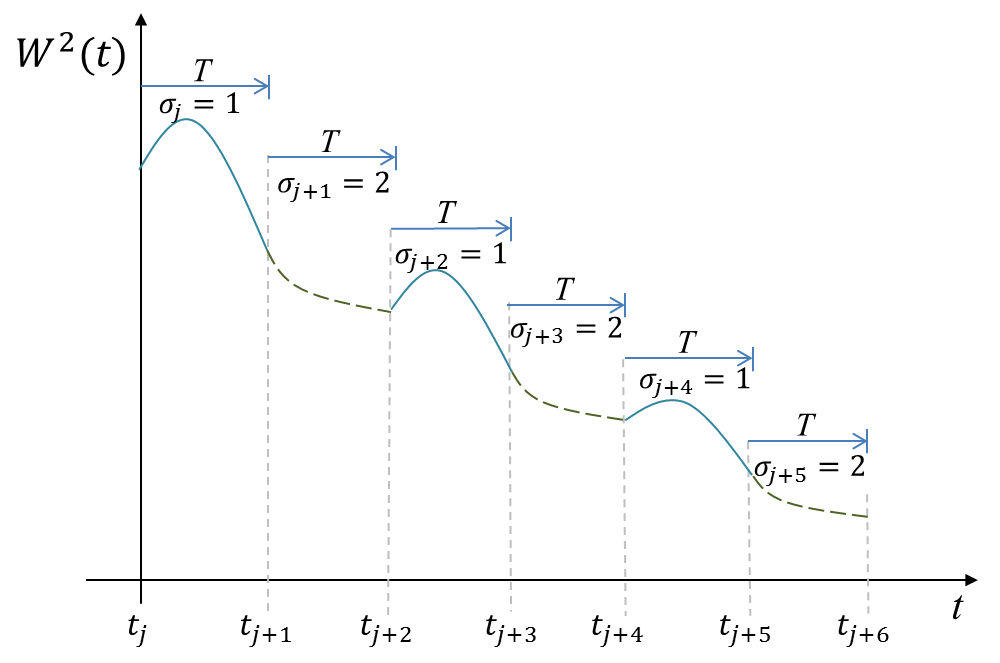}
\vspace{-0.2in}
\caption{Optimal Switching Strategy for the jump linear system using Receding Horizon Framework}
\label{fig:jump}
\end{center}
\end{figure}  

\subsection*{Stability Issues}
The reason we choose time $t_{j-1}$ and $t_{j+1}$ for piecewise monotone decreasing condition in \eqref{stab_cond} is as follows. Switching takes place at every time instance $t_j$. Between time $k\in[t_{j-1},t_j)$, there is no switching. At the end of the horizon $T$, which is at time $t_j$, we can compute the next optimal switching sequence for time $k\in[t_j,t_{j+1})$ using \eqref{Jopt2}-\eqref{stab_cond}. Since the individual subsystem is Schur stable, there is no stability problem if there is no switching. However, if jump occurs, there may be a bump in the state trajectory, and hence in $W^2$ right after the switching. This may cause instability of the jump linear system. Therefore, the constraint \eqref{stab_cond} which is sufficient condition for the stability should be enforced. The following lemma and theorem prove the stability of jump linear systems in the context of mean square sense under the receding horizon framework.

\begin{lemma}\label{piecewise}
For jump linear systems with the receding horizon framework \eqref{Jopt2}, $W^2(t_j)$ converges to zero under the constraint \eqref{stab_cond}, where $t_j$ is jump time.
\end{lemma}
\begin{proof}
For piecewise monotone decreasing sequence $W^{2}(\cdot)$, $\exists n_0\in \mathbb{Z}^+$ such that $W^{2}(t_{n_{0}}) < N$ and $N$ is any arbitrary positive real number $\mathbb{R}^{+}$. By the monotone decreasing condition above, for all $n>n_0$, $W^{2}(t_n)<N$. Since $N$ is any arbitrary positive real number $\mathbb{R}^{+}$ and the lower bound of $W^{2}(\cdot)$ is 0, $W^{2}(t_j) \rightarrow 0$ as $j \rightarrow \infty$.
\end{proof}

Lemma \ref{piecewise} proves piecewise convergence of $W^{2}$ under the constraint given in \eqref{stab_cond}. Although $W^{2}(t_j)$ converges to zero, it does not necessarily guarantee no oscillation at time $k\in[t_j, t_{j+1})$. Therefore, we have to show that if $W^{2}(t_j)\rightarrow0$, then $W^2(k)$ is also zero for all $t\in[t_j, t_{j+1})$. The following lemma proves the above argument.

\begin{lemma}\label{no_osc}
Once $W^{2}(t_j) = 0$ at time $t_j$, then $W^{2}(k)$ is always zero for all $ k \geq t_j$.
\end{lemma}
\begin{proof}
From \eqref{staticW}, in order for $W^{2}(t_j)$ to be zero, both mean and covariance have to be zero. According to \eqref{mu_update} and \eqref{sigma_update}, used for updating mean and covariance, they remain zero for all $k\geq t_j$ once they become zero.
\end{proof}

Using Lemma \ref{piecewise} and Lemma \ref{no_osc}, following theorem shows the m.s stability of jump system under the proposed switching policy.
\begin{theorem}
Jump linear systems in \eqref{3} under the receding horizon framework \eqref{Jopt2}-\eqref{stab_cond} is m.s. stable.
\end{theorem}
\begin{proof}
By Proposition \ref{m.s.stable}, the system is m.s. stable if and only if $W(\cdot)=0$. From Lemma \ref{piecewise} and Lemma \ref{no_osc}, it is shown that $W^{2}(\cdot)$ converges to zero, and hence $W(\cdot)$ also converges to zero.  Therefore, jump linear system in \eqref{3} is m.s. stable.
\end{proof}

\subsection*{Horizon Length Issues}
Primbs \textit{et al.}\cite{primbs2000receding} have shown the unified framework between pointwise min-norm($T=0$), optimality($T=\infty$), and receding horizon $T$. The horizon length $T$ can vary according to available time for online computation and in general, we can attain better results for longer horizon length $T$. However, unlike receding horizon control, longer horizon length $T$ for optimal switching does not imply better performance of jump systems. The effect of different receding horizon length $T$ in optimal switching can be analysed as follows.
\begin{enumerate}
\item \textit{Pointwise minimum ($T=0$)}:
When $T=0$, the solution of optimal switching problem is obtained by solving \eqref{Jopt2}-\eqref{stab_cond} with $T=0$. This is equivalent to finding pointwise minimum of $W^{2}_{\sigma_k}(k)$ at every time $k$. However, since there is no prediction for the future behaviour of the system, pointwise minimum does not guarantee the optimal switching of jump systems. Therefore, it may cause worse performance than good or even poor controller itself without switching.
\item \textit{Infinite horizon ($T=\infty$)}:
In case of infinite horizon, the optimal switching problem is trivial. By assumption \ref{assump1}, switching does not occur over infinite horizon. Therefore, the solution of the optimal switching is to choose single mode which achieves the minimum area of $W^2$ from time $k=0$ to $\infty$.
\end{enumerate}

From the above fact, receding horizon length $T$ should be $0< T < \infty$. However, there is no guideline for the optimal horizon length $T$. One necessary condition for $T$ is that it has to be chosen to satisfy the stability constraint in \eqref{stab_cond}. For instance, if jump occurs at time $t_j$ and as a result there might be a bump right after the switching, then the constraint\eqref{stab_cond} may not be satisfied for short $T$. Therefore, we can set the receding horizon length $T$ as follows.

\begin{theorem}\label{theorem4.2}
For optimal switching problem with receding horizon framework in \eqref{Jopt2}-\eqref{stab_cond}, the receding horizon length $T$ has to be set to satisfy stability constraint \eqref{stab_cond} and such that,
\begin{eqnarray}
T \geq \tau_j := t_{j+1} - t_{j-1}\label{tau}
\end{eqnarray}
where $\tau_j$ is updating time interval for receding horizon, and there always exist $\tau_j$ satisfies stability constraint \eqref{stab_cond} under the assumption that each dynamics is Schur stable.
\end{theorem}
\begin{proof}
From Assumption \ref{assump1} and that individual systems are Schur stable, there is no switching in fixed horizon $T$. For linearly stable system, which is globally uniformly asymptotically stable, there exists time $t_{j+1}$ such that $\|x(t_{j+1})\| < \|x(t_{j-1})\|$ for  $t_{j+1}>t_{j-1}$. By taking square and expectation for both side of above equation, we get $W^2(t_{j+1}) < W^2(t_{j-1})$. Therefore, stability constraint \eqref{stab_cond} is satisfied with some positive definite function $\epsilon_j$.
\end{proof}

Note that the horizon length $T$ is not necessarily to be constant. For each different jump time $t_j$ we can set a different horizon length $T$, satisfying the condition given in Theorem \ref{theorem4.2}.
 
\subsection*{Complexity Issues}
Two problems associated with the original optimal switching problem \eqref{Jopt} give rise to complexity issues. First, infinite time causes infinite size in total possible numbers of switching. Second, even if the switching is finite and hence \eqref{Jopt} is equivalent to tree-search problem\cite{dakin1965tree}, the computational complexity to solve this problem is NP-complete\cite{arnborg1987complexity}.

However, the optimal switching with receding horizon framework in \eqref{Jopt2}-\eqref{stab_cond} enable us to simplify the problem. Once the horizon length $T$ satisfying \eqref{stab_cond} is obtained, then the solution of optimal switching problem is same with choosing $\min\{W_1^2,W_2^2,\cdots,W_m^2\}$, where $m$ is total number of modes. Hence, this is same with sorting problems, where computational complexity is $O(n\log n)$ in general. As a consequence, optimal switching with receding horizon can be solved fast enough for online computation.


\section*{EXAMPLES}
\subsection*{Jump Linear System with Five different modes dynamics}
Consider a following discrete-time jump linear system.
\begin{align*}
x(k+1) = A_{\sigma}x(k), \qquad \sigma\in\{1,2,\hdots,5\}
\end{align*}
This system has a five different mode dynamics given by
\begin{eqnarray*}
&A_1 = \begin{bmatrix}
		  1.01 & -0.17 \\
		  0.32 & -0.48
\end{bmatrix},\:
A_2 = \begin{bmatrix}
0.06 & 0.80\\
0.01 & -0.77
\end{bmatrix},\:
A_3 = \begin{bmatrix}
0.72 & 0.48\\
0 & 0.55
\end{bmatrix},\\
&A_4 = \begin{bmatrix}
-0.33 & -0.65\\
-0.46 & 0.69
\end{bmatrix},\quad
A_5 = \begin{bmatrix}
-0.13 & 0.12\\
-1.33 & -1.05
\end{bmatrix}
\end{eqnarray*}

In addition, we assume that initial state has an uncertainty represented by Gaussian PDF with mean $\mu_0$ and covariance $\Sigma_0$ as follows.
\begin{align*}
\mu_0=[5,5]^{\top},\qquad
\Sigma_0=\begin{bmatrix}
2.25 & 0\\
0 & 2.25
\end{bmatrix}
\end{align*}.

For this system with given Gaussian initial state PDF, we aim to design a switching sequence that attains the optimality of the system performance. The simulation results are shown in Fig. \ref{fig:ex1(a)}. The cross mark represents the mode that is used at a specific switching sequence. According to this result, jump system shows the fastest convergence to the origin under the proposed receding horizon framework. 
The spectral radius for individual mode dynamics are $\rho(A_1) = 0.97$, $\rho(A_2) = 0.78$, $\rho(A_3)=0.72$, $\rho(A_4)=0.93$, and $\rho(A_5)=0.82$. From this, we know that $A_3$ dynamics converges to the origin faster than other mode dynamics. However, since other mode dynamics shows good performance in the beginning, it is desirable to use that mode dynamics initially. This intuition coincide with the optimal switching results as shown in Fig.\ref{fig:ex1(a)}. Additionally, total $W^2$ area that stands for the system performance, is depicted in Fig.\ref{fig:ex1(b)} to compare the performance between different mode dynamics and a jump system. In Fig.\ref{fig:ex1(b)}, $A_2$ mode shows the minimal $W^2$ area between individual dynamics without switching. The jump system using optimal switching synthesis shows about 3.5 times less $W^2$ area compared to $A_2$ mode that attains the best performance between the individual mode dynamics. In this example, the optimal switching synthesis provided in this paper shows the best performance and beats any other mode dynamics without switching.

\begin{figure}[!ht] 
\begin{center}
\subfigure[$W^2$ distance and optimal switching sequence $\sigma$]{\includegraphics[width=0.45\textwidth]{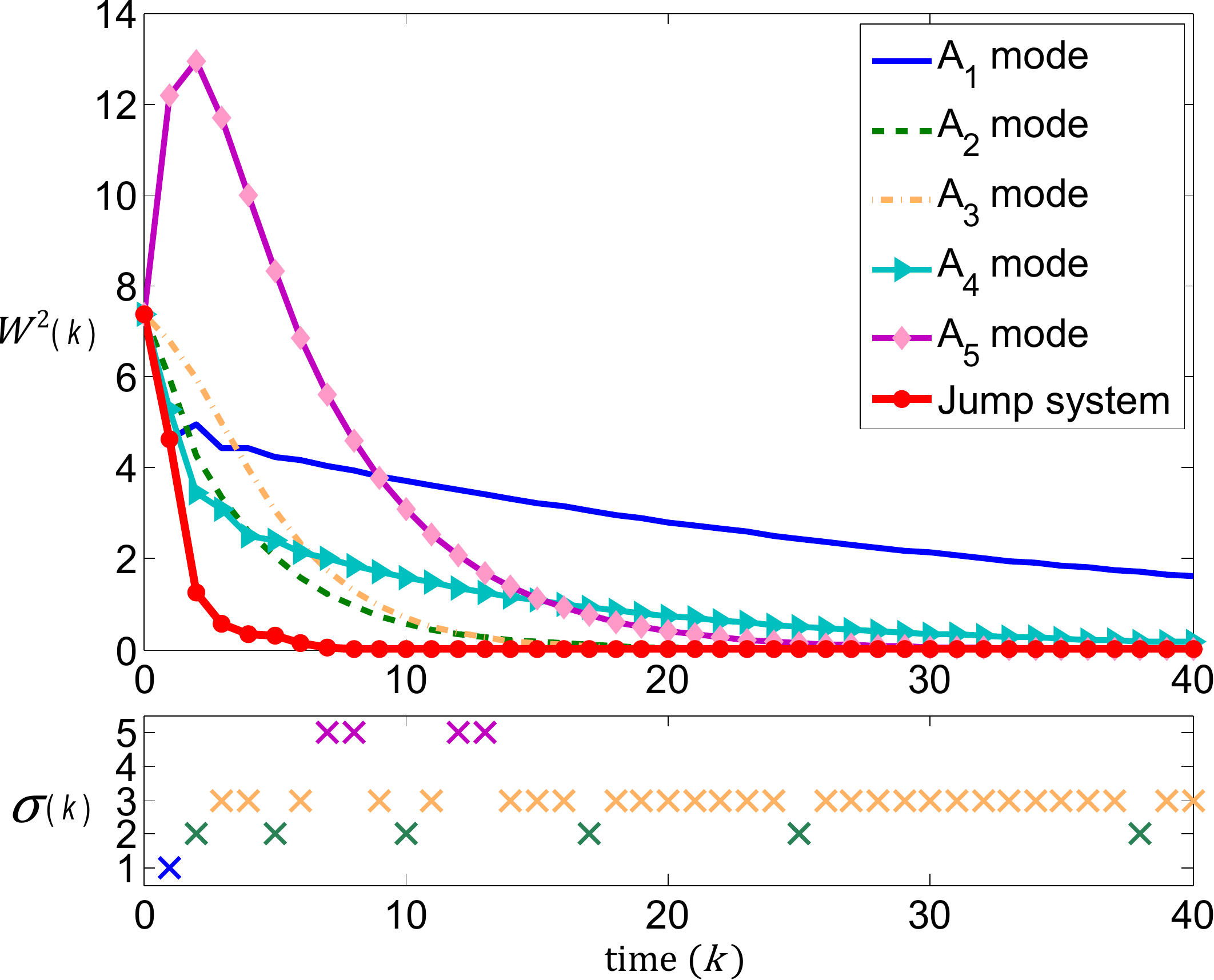}\label{fig:ex1(a)}}
\subfigure[Total area of $W^2$ for each mode and Jump system]{\includegraphics[width=0.42\textwidth]{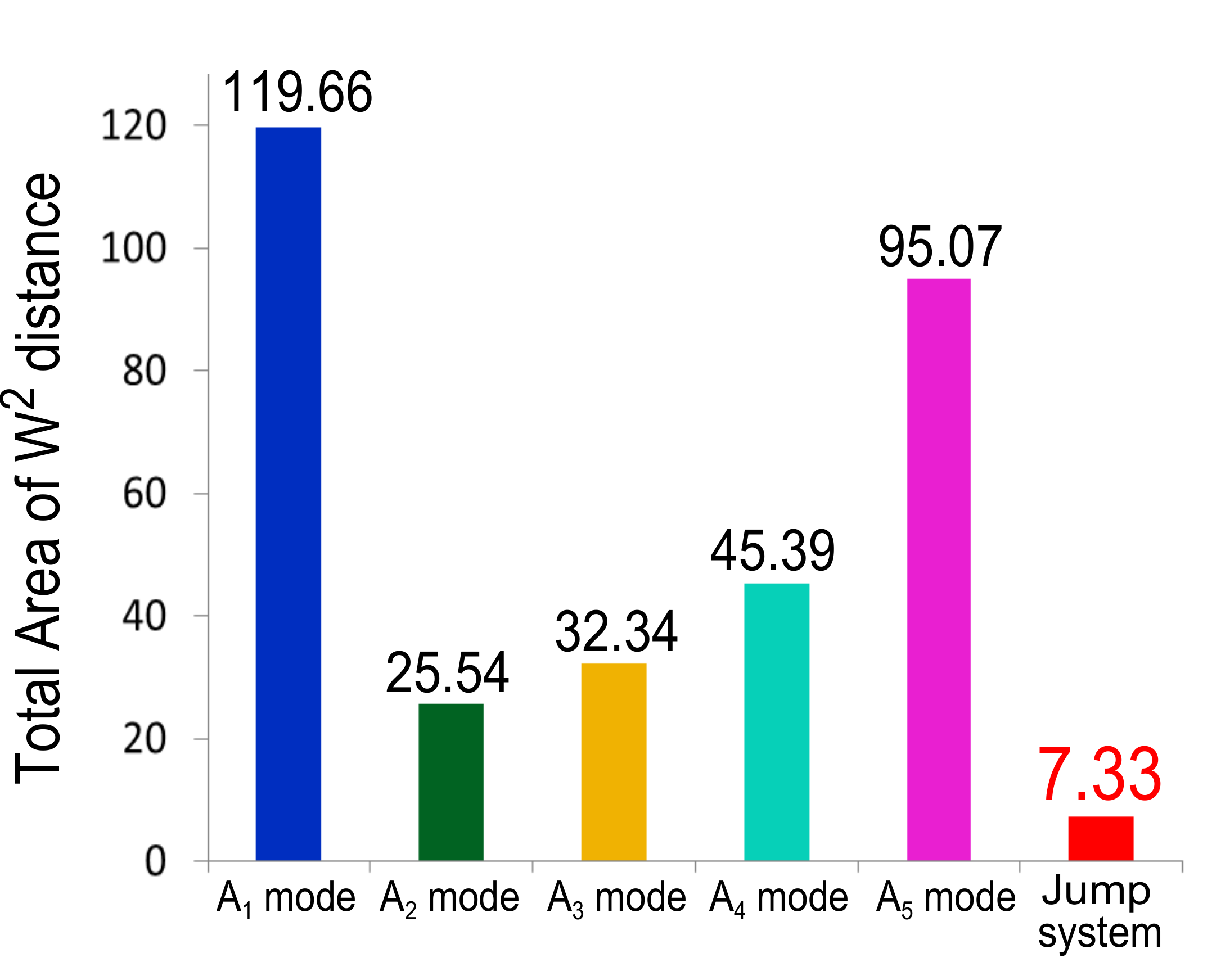}\label{fig:ex1(b)}}
\caption{Simulation results of Optimal Switching Synthesis for the Jump Linear System that has 5 different modes}
\label{fig:ex1}
\end{center}
\end{figure}

\subsection*{Linearized Quadrotor Dynamics with Two Controllers}
Here we consider 6-state linearized nonlinear quadrotor dynamics.  The first controller ($C_{High}$)  provides higher performance by commanding aggressive control actions and is designed using full-state feedback.  The second controller is a lead-lag compensator ($C_{Low}$) which provides poorer performance by commanding less aggressive control actions. Implementation of $C_{High}$ requires more computational time and consumes more energy (batttery) and $C_{Low}$ is resource economical in terms of both CPU time and energy usage. More details about this controller can be found in \cite{kottenstette2009digital}. In this example, we want to design the optimal switching sequence using both $C_{high}$ and $C_{low}$ to obtain better performance.

The states of the quadrotor are $x = [\phi,\theta,\psi, p, q, r]^\top$ and nonlinear dynamics is given by
\begin{align*}
&\dot{p} = \frac{qr(I_{yy}-I_{zz}) + qJ_{r}\Omega_{r} + bl(-\Omega_{2}^{2}+\Omega_{4}^{2})}{I_{xx}},\\
&\dot{q} = \frac{pr(I_{zz}-I_{xx}) - pJ_{r}\Omega_{r} + bl(\Omega_{1}^{2}-\Omega_{3}^{2})}{I_{yy}},\\
&\dot{r} = \frac{pq(I_{xx}-I_{yy}) + d(-\Omega_{1}^{2}+\Omega_{2}^{2}-\Omega_{3}^{2}+\Omega_{4}^{2})}{I_{zz}},
\end{align*}

\[
\begin{pmatrix}
\dot{\phi}\\
\dot{\theta}\\
\dot{\psi}\\
\end{pmatrix}
=
\begin{pmatrix}
1	&	\sin(\phi)\tan(\theta)	& \cos(\phi)\tan(\theta)\\
0	&	\cos(\phi)				& -\sin(\phi)\\
0	&	\sin(\phi)\sec(\theta)	& \cos(\phi)\sec(\theta)
\end{pmatrix}
\begin{pmatrix}
p\\q\\r
\end{pmatrix},
\]
where symbols are defined in Table \ref{table_quad}. 
\begin{table}[h]
\begin{center}
\caption{Nomenclature for Quadrotor Dynamics}\label{table_quad}
\label{table_quadrotor_nomenclature}
  \begin{tabular}{|c|c|c|c|}\hline
  Symbol & definition & Symbol & definition\\\hline
  $\phi$ & roll angle & $p$ & roll rate\\  
  $\theta$ & pitch angle & $q$ & pitch rate\\
  $\psi$ & yaw angle & $r$ & yaw rate\\
  $I_{xx,yy,zz}$ & body inertia &   $J_{r}$ & rotor inertia\\
  $b$ & thrust factor &  $d$ & drag factor\\ 
  $l$ & lever & $\Omega_{r}$ & rotor speed
  \\\hline
  \end{tabular}
\end{center}
 \end{table}

Linearized quadrotor dynamics is obtained by linearizing the nonlinear equations of motion about hover. Two continuous-time closed-loop systems $A_1$ and $A_2$ are discretized with sampling time $0.01$s. The switching policy determines the sequence for $\sigma$, which is deterministic.

The initial condition uncertainty is assessed with respect to initial condition uncertainty given by Gaussian PDF $\mathcal{N}(\mu_{0},\Sigma_{0})$, with $\mu_{0}=[0.5,-1.5,-5,0.1,0.2,0.1]^{\top}$ and $\Sigma_0=0.0225\times I_{6\times 6}$, where $I_{6\times 6}$ is the $6\times 6$ identity matrix. The control objective is to maintain hover, which corresponds to equilibrium state $x_{eq}=[0,0,0,0,0,0]^{\top}$.

\begin{figure}[!ht] 
\begin{center}
\subfigure[$W^2$ distance and optimal switching sequence $\sigma$]{\includegraphics[width=0.45\textwidth]{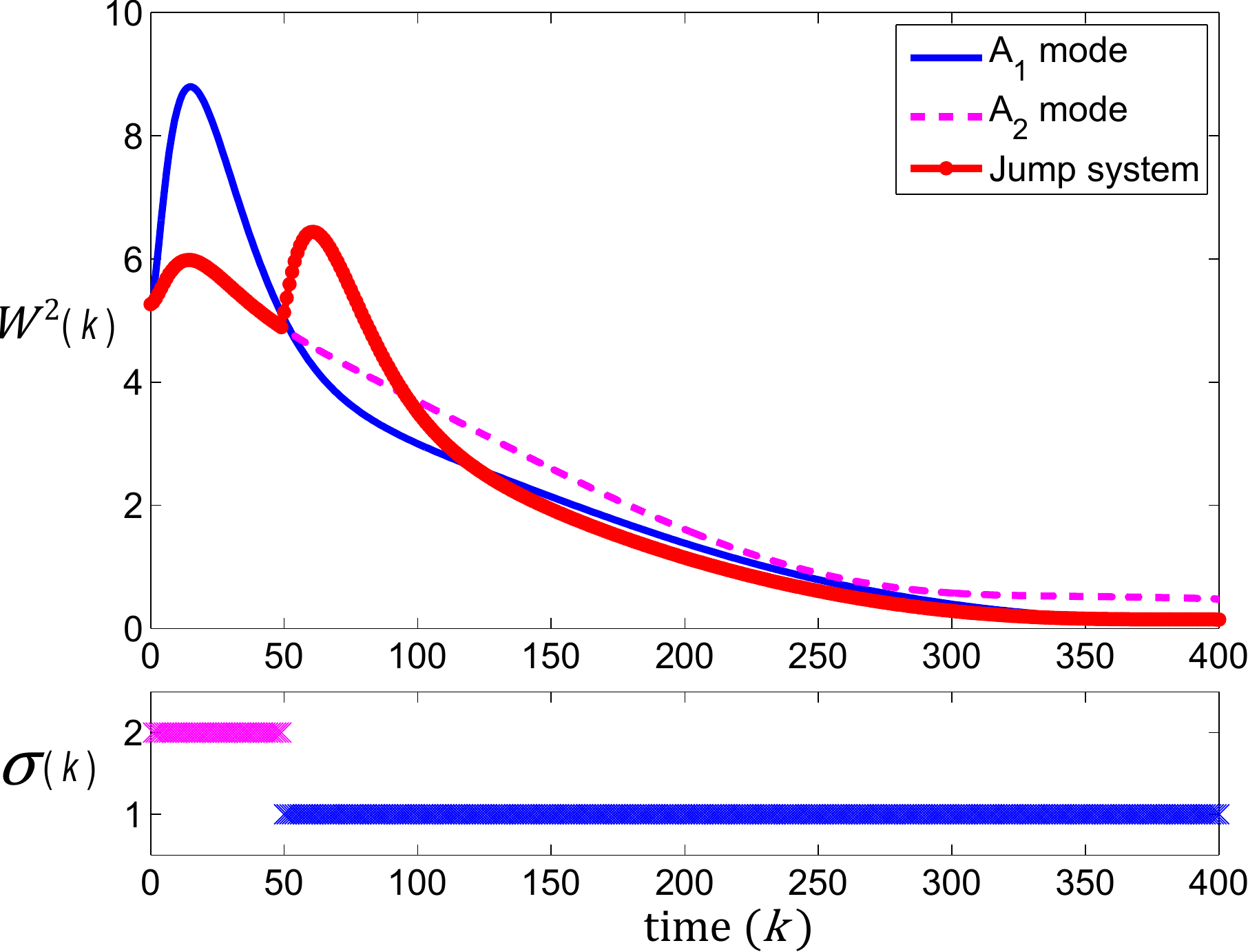}\label{fig:ex2(a)}}
\subfigure[Total area of $W^2$ for each mode and Jump system]{\includegraphics[width=0.42\textwidth]{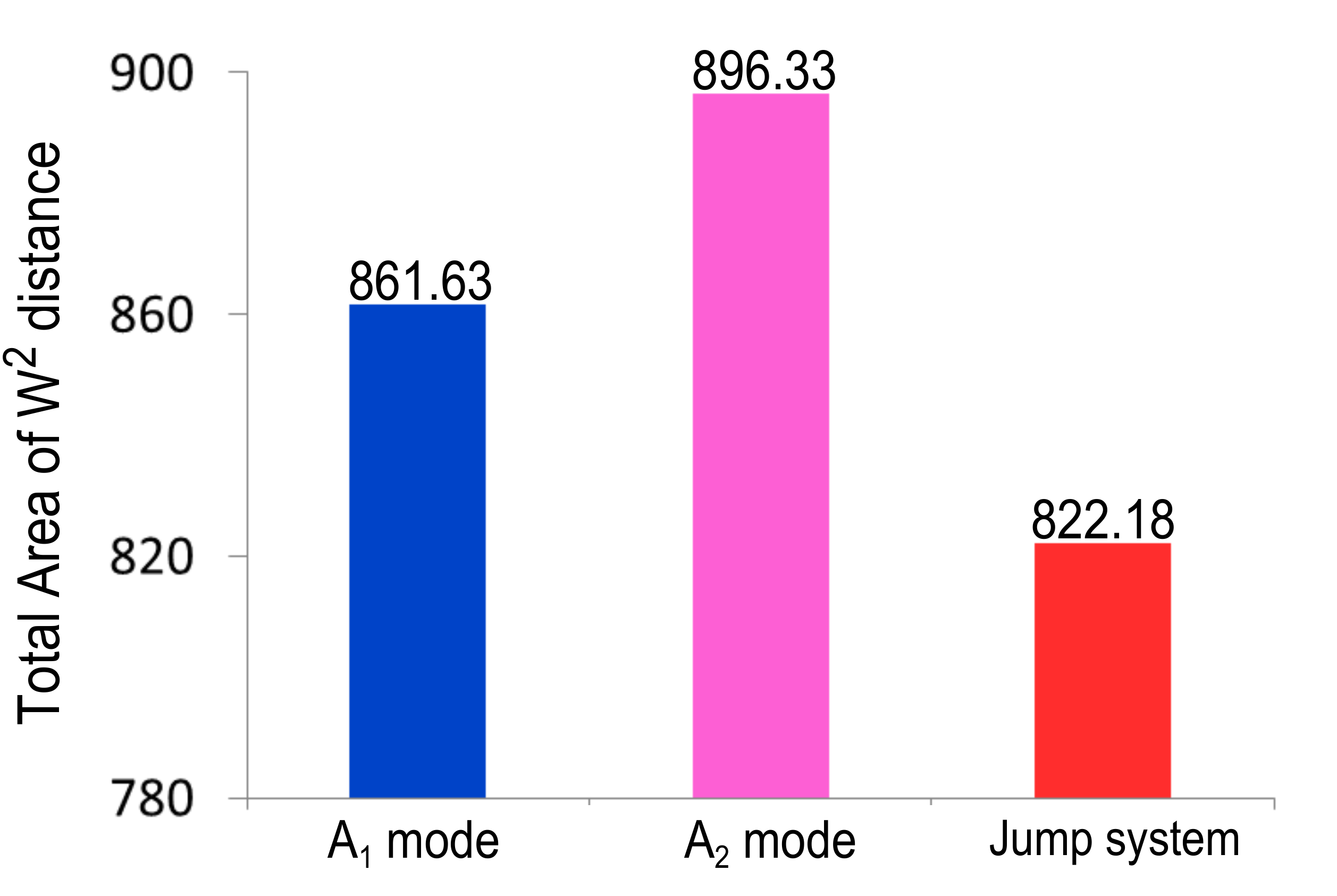}\label{fig:ex2(b)}}
\caption{Simulation results of Optimal Switching for Linearized quadrotor dynamics}
\label{fig:ex2}
\end{center}
\end{figure}

Fig. \ref{fig:ex2} presents the result of switching synthesis using proposed method in this paper. From the beginning in Fig. \ref{fig:ex2(a)}, $A_1$ dynamics shows large elevation in $W^2$ distance while $A_2$ does not. As a result, the optimal switching with receding horizon selects $A_2$ dynamics. However, after $k=50$, an optimality is obtained by switching to $A_1$ via optimal switching synthesis proposed in this paper. Fig. \ref{fig:ex2(b)} presents the performance of each mode in terms of total $W^2$ area.  It is clear that the lowest area, which is the best performance, can be attained by switching.

%
%


\section*{CONCLUSION}
In this paper, we proposed the optimal switching synthesis for jump linear systems with Gaussian initial state uncertainty. The Wasserstein metric that defines a distance between PDFs was adopted to measure both the performance and the stability of the jump linear system. We showed that the optimality of the system performance can be obtained by synthesizing switching laws via minimization of objective function expressed in terms of Wasserstein distance. Also, the mean square stability of the jump linear system was guaranteed under the proposed switching synthesis. The efficiency and the usefulness of the proposed methods were demonstrated by examples.

\bibliographystyle{asmems4}

\begin{acknowledgment}
This research was supported through National Science Foundation award \#1349100, with
Dr. Almadena Y. Chtchelkanova as a program manager.
\end{acknowledgment}

%
% ...


\bibliography{DSCC2014}

%

\end{document}